\documentclass{article}

\usepackage[textwidth=15cm,textheight=21.7cm]{geometry}
\usepackage{graphicx}
\usepackage[square,numbers]{natbib}
\usepackage{graphicx}
\usepackage{physics}
\usepackage{amsmath}
\usepackage{amsfonts}
\usepackage{amsthm}
\usepackage{siunitx}
\usepackage{bbm}
\usepackage{bm}
\usepackage[title]{appendix}

\newtheorem{theorem}{Theorem}

\newcommand{\hilb}[0]{%
  \mathcal{H}%
}

\newcommand{\tensp}[1]{%
  \mathbin{\mathop{\otimes}\displaylimits_{#1}}%
}

\begin{document}

\title{On the equivalence between quantum and random walks on finite graphs}

\author{Matheus G. Andrade \and
        Franklin Marquezino \and
        Daniel R. Figueiredo \\
        Department of Computer and System Engineering COPPE \\
        Federal University of Rio de Janeiro (UFRJ), Brazil \\
        \{guedesm, franklin, daniel\}@cos.ufrj.br
}



\maketitle

\begin{abstract}
Quantum walks on graphs are ubiquitous in quantum computing finding a myriad of applications. Likewise, random walks on graphs are a fundamental building block for a large number of algorithms with diverse applications.
While the relationship between quantum and random walks has been recently discussed in
specific scenarios, this work establishes a formal equivalence between the processes on arbitrary
finite graphs and general conditions for shift and coin operators. It requires empowering random
walks with time heterogeneity, where the transition probability of the walker is non-uniform and time dependent.
The equivalence is obtained by equating the probability of measuring the quantum walk on a given node
of the graph and the probability that the random walk is at that same node, for all nodes and time steps.
The result is given by the construction procedure of a matrix sequence for the random walk that
yields the exact same vertex probability distribution sequence of any given quantum walk, including the scenario with multiple interfering walkers.
Interestingly, these matrices allows for a different simulation approach for quantum walks where node samples respect
neighbor locality and convergence is guaranteed by the law of large numbers, enabling efficient (polynomial) sampling of quantum graph trajectories (paths).
Furthermore, the complexity of constructing 
this sequence of matrices is discussed in the general case.
\end{abstract}

\section{Introduction}
Quantum walks on graphs are a prominent area of research in quantum computing
inspired to be the quantum analogue of classical random walks~\cite{aharonov2001quantum, aharonov1993quantum}.
As with random walks, quantum walks have proven to be an insightful tool for designing
quantum algorithms, culminating on efficient solutions for problems such as element distinctness
\cite{ambainis2007quantum}, marked-vertex searching~\cite{Magniez2004search} and Hamiltonian simulation~\cite{berry2015hamiltonian}. Among its marvelous capabilities,
quantum walks were shown to perform universal quantum computation for both continuous-~\cite{childs2009universal} and discrete-time~\cite{lovett2010universal} models.
Extensive surveys covering multiple aspects of quantum walks can be found in the literature~\cite{kempe2009quantum, portugal2013quantum, venegas2012quantum}. 

A few discrete-time models for quantum walks have shown increased community interest over the past years
\cite{aharonov2001quantum, portugal2016staggered, szegedy2004walks}. The coined model works on an extended
Hilbert space which codifies both graph vertices and walker direction and has pioneered 
discrete-time models~\cite{aharonov2001quantum}, where the coin space was
introduced to allow unitary evolution. The later
Szegedy model~\cite{szegedy2004walks} performs quantization over a bipartite Markov chain. In this model, a reflection based operator is constructed once the transition probabilities to cross the bipartite sets are defined. 
The operators of the Szegedy model have a well described spectra and its properties
are mainly derived from spectral analysis. The staggered model~\cite{portugal2016staggered} is based on
graph tesselations and generalizes the bipartite construction of the Szegedy walk. This work focuses
on the early coined model.

The case of multiple walkers has also been investigated in different contexts. As with the single quantum walker, the interacting multi-walker model was also shown to be universal for quantum computing~\cite{childs2013universal}. Non-interacting multi-walker models on arbitrary graphs have been treated generically, with proposed physical implementations~\cite{rohde2011multi}. The two-walker case was specifically analyzed, leading to interesting results~\cite{siloi2017walks, vstefanak2006meeting, xue2012walkers}.

In the classical realm, random walks on graphs~\cite{lovasz1993random} have been extensively used to drive the design
of classical algorithms for a myriad of problems in diverse areas of computing, ranging from 
sampling~\cite{hastings70mcmc} to user recommendation~\cite{page1999pagerank}.
Most applications of random walks assume time homogeneity, which implies that the walker behavior, as it
moves on the graph, does not change over time. Time
homogeneity favors analytical tractability and important known results have been derived under this restriction, such as the conditions for time convergence of the probability distribution~\cite{haggstrom2002finite}.
On the other hand, non-homogeneity, or time-dependence,
has been explored on particular niches, such as the celebrated
Simulated Annealing meta-heuristic for optimization problems~\cite{kirkpatrick1983optimization}.

The connection between quantum and random walks has been investigated and it is clear that homogeneous random walks cannot match quantum walks on arbitrary graphs. However, it has been shown that the evolution of quantum walks on infinite lines are partially described by time-homogeneous Markovian processes~\cite{romanelli2004quantum}. Its probability evolution can be expressed as a time-independent Markov process with an additional interference term. This separation method was further used to construct a master equation for the global chirality distribution (GCD) of the quantum walk~\cite{romaneli2010gcd}, showing a convergence behavior of homogeneous Markovian processes for the GCD. In addition, a relationship between the walk dimension of both processes was explored through the use of renormalization-group analysis (RG) to evaluate scaling factors of the quantum walk limiting distribution~\cite{boettcher2015relation}. This analysis allows for the calculation of the walk dimension for quantum walks on some non-trivial graphs and has lead to the conjecture that the number of walk dimensions for the quantum case is half of that of the random walk, a well known result in the case of homogeneous lattices~\cite{boettcher2015relation}.

In the search for their equivalence, a recent work has shown that non-homogeneous random walks can yield probability 
sequences identical to the probabilities of quantum walks on the infinite integer line~\cite{montero2017quantum}.
In this context, an analysis was carried out to generate a given desired distribution sequence over the integers with
time- and site-dependent discrete-time coined quantum walks and non-homogeneous random walks.
The matching is performed by constructing a random walk with time-varying probabilities that has the same distribution sequence of a Hadamard-coined quantum walk on the infinite line.

A different perspective is the Quantum Stochastic Walk (QSW) model, a generalization of both quantum and
random walks which accounts for non-unitary transformations~\cite{whitfield2010quantum}. Using the 
formalism of density matrices, a super operator is constructed to perform both Hamiltonian (coherent)
and stochastic evolution based on the Kossakowski-Lindblad master equation. The walk behavior over a graph
is achieved upon connectivity restrictions on the terms that map the states of the system. Depending on 
how such terms are chosen, the behavior of both classical and quantum walks can be obtained, as well as
the behavior of a more general quantum stochastic process not captured by either of them. However, QSW has no bearing on the equivalence of random and quantum walks.

This article focus on the connection between unitary discrete-time coined quantum walks and random walks
on finite graphs, and formally proves that the probability evolution of any quantum walk
can be matched exactly by a time-dependent random walk on the same underlying graph. This connection stems
from the locality property of both random and quantum walks. Our main 
contribution is a prescription for the time-dependent matrix that drive the random walk dynamics in order to
produce the same probability distribution sequence of any quantum walk.
More precisely, when the random walk evolves according to these matrices, its probability distributions over the vertices are
identical to that of the quantum walk. While the sequence of matrices describing the random walk clearly depends on the graph and the quantum walk operators, the prescription is very general and requires mild assumptions, such as unitarity.

Furthermore, the equivalence is also established for the case of interacting multiple quantum walkers. The interaction model is taken to be very general, with restrictions solely on the walkers' movement. The equivalence is provided by equating the evolution of the joint probability distribution of the multiple walkers with the joint distribution of the same number of random walkers. The proof for the single-walker case is gracefully extended to the multiple walkers through arguments of unitarity. As the quantum case, the state representation for the random walk has to increase in order to accommodate all possible movements. This behavior is captured by constructing a graph in which nodes represent the current position of the walkers. The process can than be viewed as a single random walk on a much larger graph.

A direct consequence of the time-dependent matrices that provide the equivalence is the possibility to simulate a time-dependent random walk on the graph which is equivalent to its quantum walk counterpart. This simulation captures quantum behavior while generating samples that preserve neighbor locality. Differently than the commonly used quantum walk simulation procedure, the samples obtained from the random walk simulation are paths of the graph, allowing trajectories driven by the quantum behavior to be sampled.

It is worth noting that quantum walks on graphs resembles Feynman's path integral formulation for quantum mechanics~\cite{feynman2010quantum} in discrete time and space, in the sense that the probability amplitude of a discrete-time walker system at instant $t$ is described by summing up the contributions of all possible paths in the graph with length $t$ connecting the initial and final states.
In an essential way, the simulation of trajectories through random walks is a procedure for sampling paths from quantum walks following a trajectory distribution in which, for every instant $t$, the marginal vertex distribution coalesces to the quantum walk vertex distribution. This provides a powerful tool for efficient simulation of quantum walk trajectories on arbitrary graphs.



The remainder of this article is structured as follows. The notation for 
both quantum and random walks, as well as formal definitions, appears in Section~\ref{sec:representation}. The Theorem that shows how to construct the equivalent non-homogeneous random walk for any given quantum walk is stated and proved in Section~\ref{sec:equivalence}. In Section~\ref{sec:multiple}, the results are generalized for the case of multiple walkers.
The simulation of 
trajectories from the random walk matrices is treated in Section~\ref{sec:simulation}. An evaluation of the time complexity of
the procedure to construct the transition matrices appears in Section~\ref{sec:complexity}. Final remarks are drawn in Section~\ref{sec:conclusions}.

\section{Quantum and random walks}\label{sec:representation}

Let $G = (V,E)$ be a directed graph obtained from an non-directed graph by introducing two directed edges for each initial one, i.e $(u, v) \in E$ if, and only if $(v, u) \in E$. Let the sets $N^+(v) \subseteq {V}$ and $N^{-}(v) \subseteq {V}$ to denote the sets of outward and inward neighbors of $v$, respectively. 

\subsection{Quantum Walks}

A discrete-time coined quantum walk on a graph $G$ is an evolution process of a complex vector in a Hilbert space $\hilb_w \subseteq \hilb_v \tensp{} \hilb_c$ defined by the graph structure \cite{portugal2016coined}. The vertex space $\hilb_v$ has dimension $|V|$ and codifies the vertices of the graph, while the coin space $\hilb_c$ denotes the degrees of freedom of the walker movements, with dimension given by the maximum degree of the graph $D=\max\{d(v): v \in V\}$. Precisely, $\hilb_w$ is $\hilb_v \tensp{} \hilb_c$ only when $G$ is a regular graph.

Denoting $\{\ket{c}\}$ and $\{\ket{v}\}$, respectively, as the basis for the spaces $\hilb_c$ and $\hilb_v$, and letting $C_v = \{0,..., d(v) - 1\}$ be the integer set for the number of outward edges of a node $v$, the basis for $\hilb_w$ is $\{\ket{v,c}: v \in V, c \in C_v\}$. Assuming $\ket{\Psi(t)}$ is the walker wavefunction at discrete time instant $t$, the quantum walk evolution is given by the action of two unitary operators $S : \hilb_w \rightarrow \hilb_w$ and $W : \hilb_w \rightarrow \hilb_w$ on the system state vector as
\begin{align}
    & \ket{\Psi(t + 1)} = SW \ket{\Psi(t)}. \label{eq:qwalk}
\end{align}
In this work, we assume that both $S$ and $W$ may vary with time, although the dependence will be omitted in order to simplify notation.

\subsubsection{The coin operator}

The coin operator ($W$) acts on the degrees of freedom of the walker. The most general coin operator is given by
\begin{align}
    & W = \sum_{v \in V}\dyad{v}{v} \tensp{} \sum_{j \in C_v}\sum_{k \in C_v} \dyad{j}{k} w_{vjk}, \label{eq:coin_op}
\end{align}
which is the form considered throughout this work. This operator is responsible for mixing the amplitude of a given state $\ket{v,c}$ with states $\ket{v, c'}$ such that $c, c' \in C_v$, i.e degrees of freedom of the same vertex, through weights $w_{vc'c}$. The mixing behavior is enlightened when one observes the action of $W$ on a generic state vector $\ket{v,c}$

\begin{align}
    & W\ket{v,c} = \sum_{u \in V}\sum_{i,j \in C_v}w_{uji} \dyad{u}{u} \tensp{} \dyad{j}{i} \ket{v,c} = \sum_{j \in C_v}\ket{v, j}w_{vjc}. \label{eq:amp_mixture}
\end{align}


For $W$ to be unitary, one must impose conditions on the complex values of $w_{vjc}$. In particular, the product operator $WW^{\dagger}$ is given as
\begin{align}
    &  W^{\dagger}W = \sum_{v \in V}\sum_{j \in C_v}\sum_{i \in C_v}\sum_{l \in C_v}\sum_{k \in C_v} w^{*}_{vij}w_{vlk} \dyad{v}{v} \tensp{} \dyad{j}{i}  \tensp{} \dyad{l}{k}, \\
    & W^{\dagger}W = \sum_{v \in V}\sum_{j \in C_v}\sum_{i \in C_v}\sum_{k \in C_v} w^{*}_{vij}w_{vik} \dyad{v}{v} \tensp{} \dyad{j}{k} \label{eq:coin_coef}
\end{align}
and the coefficients of the right hand side of Equation \ref{eq:coin_coef} must obey
\begin{align}
    & \sum_{i \in C_v} |w_{vik}|^2 = 1 : v \in V, \text{ and } \label{eq:coin_eq} \\
    & \sum_{i \in C_v}\sum_{j \in C_v}\sum_{k \neq j} w^{*}_{vij}w_{vik} = 0: v \in V.\label{eq:coin_diff}
\end{align}

Two coin operators which will be important further ahead are the Hadamard and the Grover operators. The $D\text{-dimensinoal}$ Hadamard operator $H_D$ can be constructed for Hilbert spaces with dimension of the form $D = 2^k$, for $k \in \{1,2,...\}$. Its formal definition
is given by
\begin{align}
    & H_D = H_{\frac{D}{2}} \tensp{} H_2 \label{eq:hadamard_op}
\end{align}
where
\begin{align}
    & H_2 = \frac{1}{\sqrt{2}}\begin{pmatrix}
    1 & 1 \\
    1 & -1
    \end{pmatrix}. \label{eq:had2}
\end{align}{}
On the other hand, the Grover operator can be defined for Hilbert spaces with arbitrary dimension, being formally represented as
\begin{align}
    & G = \sum_{c = 0}^{D - 1}\sum_{c' = 0}^{D - 1}\dyad{c}{c'} - 2I, \label{eq:grover_op}
\end{align}
where $\ket{c}$ denotes vectors of the computational basis.

\subsubsection{The shift operator}

The shift, or swap, operator ($S$) acts by moving the mixed amplitudes created by the operator $W$ through outward edges. Let $\eta: V \cross C \rightarrow V$ be a mapping of vertices with its outward neighbors through an ordering of its outward edges, \textit{i.e} $u = \eta(v,c)$ is the $c\text{-th}$ outward neighbor of $v$; let $\sigma: V \cross V \rightarrow C$ be the function that maps a degree of freedom $c$ of an outward edge of a vertex $v$ with one of its inward neighbors $u$, \textit{i.e} $\sigma(u, v) = c$ is an association of the degree of freedom $c$ of $v = \eta(u, c')$ with $u$; and let $\sigma^{-1}$ to be the inverse association of this pair, \textit{i.e} $\sigma^{-1}(\eta(v,c), u) = c'$. The action of the shift operator is formally defined as

\begin{align}
    & \ket{u,c} \rightarrow \ket{\eta(u, c), \sigma(u, \eta(u, c))}. \label{eq:shift_def}
\end{align}

The functions $\eta$ and $\sigma$ can be defined in multiple ways as long as the operator remains unitary and the graph edges are respected. In fact, different definitions for these functions lead to different dynamics for the state amplitude. The action of $SW$ on a generic state vector
$$\ket{\Psi(t)} = \sum_{v \in V}\sum_{c \in C_v}\Psi(v, c, t)\ket{v,c}$$
is given by

\begin{align}
    & \ket{\Psi(t + 1)} = SW\sum_{v \in V}\sum_{c \in C_v}\Psi(v, c, t)\ket{v,c} \\
    & \ket{\Psi(t + 1)} = \sum_{v \in V}\sum_{c \in C_v}\sum_{j \in C_v}w_{vjc}\Psi(v, c, t)\ket{\eta(v, j), \sigma(v, \eta(v, j))} \\
    & \ket{\Psi(t + 1)} = \sum_{v \in V}\sum_{u \in N^{-}(v)}\sum_{j \in C_u}(\Psi(u, j, t)w_{(u,\sigma^{-1}(u,v),j)}) \ket{v, \sigma(u, v)} \label{eq:qw_evo}
\end{align}

Equation \ref{eq:qw_evo} is obtained by noting that each degree of freedom of a given vertex $v$ corresponds to exactly one neighbor of $v$ and by fixing the vertex element of the basis state vector from the summation through a variable substitution from $\eta$. Thus, the probability of seeing the walker on a given state is $\rho(v,c,t) = \abs{\Psi(v,c,t)}^2$. Since the walker states do form a basis for the state space, the total probability of seeing the walker on a given vertex is
\begin{align}
    & \rho(v, t) = \sum_{c \in C_v} \abs{\Psi(v,c,t)}^2,
    \label{eq:v_prob} 
\end{align}
which combined with Equation \ref{eq:qw_evo} leads to
\begin{align}
    & \rho(v, t) = \sum_{u \in N^{-}(v)} \abs{\sum_{j \in C_u}\Psi(u, j, t)w_{(u,\sigma^{-1}(u,v),j)}}^2. \label{eq:qw_evo_prob}
\end{align}

A common swap operator which will be mentioned on further sections of this article is the moving-shift operator, which is simply defined by the relationship
\begin{align}
    & M : \ket{u, c} \rightarrow \ket{\eta(u,c), c}. \label{eq:moving_shift}
\end{align}


\subsection{Non-homogeneous random walks}

A non-homogeneous random walk on a directed graph $G=(V,E)$ is, in essence, a diffusion process of a probability distribution over the vertices of $V$ through the edges of $E$ with time-varying transition (conditional) probabilities. Let $\pi(t) \in \mathbb{R}_{+}^{|V|}$ denote a probability vector (or a discrete probability distribution) over the set $V$ at discrete time instant $t$.
Let $p_{vu}(t) \in [0,1]$ be the transition probability for the walker to step from node $u$ to node $v$, for which holds the law of total probability
and that $p_{vu} > 0$ only if $(u,v) \in E$. The behavior of the random walk is determined by the evolution of its probability distribution given by
\begin{align}
    & \pi_v(t + 1) = \sum_{u \in N^{-}(v)} p_{vu}(t)\pi_u(t). \label{eq:rw_locality}
\end{align}

 Equation \ref{eq:rw_locality} states that the probability of a vertex at instant $t + 1$ is given by a convex combination of the probabilities of its inward neighbors, on the previous instant $t$. From this perspective, the sets of transition probabilities can be defined arbitrarily as long as the law of total probability remains valid, implying that the distributions that can be achieved by time evolution are fundamentally constrained by Equation \ref{eq:rw_locality}. This property will be denoted as the \textit{local convex evolution of probabilities}.
 
In matrix form, Equation \ref{eq:rw_locality} is represented as
\begin{align}
    & \pi(t + 1) = P(t)\pi(t), \label{eq:rw}
\end{align}
where $P(t)$ is a stochastic matrix with entries $p_{vu}(t)$ denoting the transition probability to move from vertex $u$ to vertex $v$, at instant $t$.

Note that when $\pi_u(t) = 0$, the values of transition probabilities $p_{vu}(t)$ do not contribute to the diffusion process at further times, i.e $p_{vu}(t)$ does not influence $\pi(t + k)$ for $k > 0$.

\section{Quantum walks as non-homogeneous random walks}\label{sec:equivalence}

The law of total probability and Equation \ref{eq:rw_locality} provide the starting point to establish the equivalence between quantum and random walks. From this perspective, it is necessary to define the non-homogeneous random walk that has $\pi_v(t) = \rho(v, t)$, for all $t$. A sufficient condition is the construction of the time-dependent transition matrix $P(t)$ for which $\rho(t + 1) = P(t)\rho(t)$, for all $v, t$. The existence of such sequence of matrices implies the principle of local convex evolution, in the sense of Equation \ref{eq:rw_locality}, for the full quantum walk operator $SW$, regardless of initial conditions. Theorem \ref{th:qw_locality} establishes the construction of the random walk matrix sequence.

\begin{theorem}[Quantum walk local convex evolution]\label{th:qw_locality}
For any time instant $t$, the evolution of the vertex probability of a quantum walk performed by the action of the unitary operator $SW$ is locally convex and is given by the Markovian matrix
\begin{equation}\label{eq:p_matrix}
    p_{vu}(t) =
    \begin{cases}
      \frac{\rho(v, c, t + 1)}{\rho(u, t)}, & \text{if } \rho(u, t) > 0 \text{ and } (u,v) \in E \\
      \frac{1}{d(u)}, & \text{if } \rho(u, t) = 0 \text{ and } (u,v) \in E \\
      0, & \text{otherwise}
    \end{cases}
\end{equation}
when applied to the vertex probability vector $\rho(t) \in \mathbb{R}^{|V|}_{+}$, where $c = \sigma(u, v)$, and such that $\rho(t + 1) = P(t)\rho(t)$
\end{theorem}

\begin{proof}
To completely prove the claim, it is necessary to show that the following three properties hold for $P$:

\begin{center}
    \begin{enumerate}
        \item $ 0 \leq p_{vu}(t) \leq 1$ for every $u, v \in V$;
        \item $ \sum_{v \in N^{+}(u)}^{} p_{vu}(t) = 1$ for each $v \in V$;
        \item $ \rho(v, t + 1) = \sum_{u \in N^{-}(v)} p_{vu}(t)\rho(u, t)$ for each $v \in V$.
    \end{enumerate}
\end{center}

Whenever $\rho(u, t) = 0$, choosing $p_{vu}(t) = \frac{1}{d(u)}$ avoids division by zero and assures the first and the second conditions. Since $p_{vu}(t)\rho(u, t) = 0$ for this particular case, the task is to show that the three conditions hold for  $\rho(u, t) > 0$. Note that $p_{vu}(t)$ could be chosen arbitrarily, as long as the $u\text{-th}$ column of $P$ respected conditions $1$ and $2$. Uniform weights were chosen for simplicity.
Using Equation \ref{eq:qw_evo_prob}, and taking $c \in C_v$, $c = \sigma(u, v)$ and $c' = \sigma^{-1}(u,v)$, one has:

\begin{align}
    & \rho(v, c, t + 1) = \abs{\sum_{j \in C_u}\Psi(u, j, t)w_{uc'j}}^2 \label{eq:qw_v_prob} \\
    & p_{vu}(t) = \frac{\abs{\sum_{j \in C_u}\Psi(u, j, t)w_{uc'j}}^2}{\rho(u, t)}. \label{eq:qw_beta_frac}
\end{align}
The numerator on the right-hand side of Equation \ref{eq:qw_beta_frac} can be thought of as the result of the inner product between the vectors $\ket{\Psi^*(u,t)}$ and $\ket{W_u}$ with $j\text{-th}$ coordinates respectively given by $\ket{\Psi^{*}(u,t)}_j = \Psi^{*}(u,j,t)$ and $\ket{W_u}_j = w_{uc'j}$, $j \in C_u$. By the Cauchy-Schwarz inequality
\begin{align}
    & \abs{\braket{\Psi^{*}(u, t)}{W_u}}^2 \leq \braket{\Psi^*(u, t)}{\Psi^*(u, t)} \braket{W_u}{W_u}.
\end{align}
Since $\braket{W_u}{W_u} = 1$ due to the unitarity of $W$ (Equation \ref{eq:coin_eq}),
\begin{align}
    & \abs{\sum_{j \in C_u}\Psi(u, j, t)w_{(uc'j)}}^2 \leq \rho(u, t) \label{eq:beta_small1_proof}
\end{align}
implies that $p_{vu}(t) \leq 1$. As both the numerator and the denominator of Equation \ref{eq:qw_beta_frac} are positive, $p_{vu}(t) \geq 0$, proving property $1$.

Furthermore, the numerator of the sum of conditional probabilities
\begin{align}
    & \sum_{v \in N^{+}(u)} p_{vu}(t) = \frac{\sum_{v \in N^{+}(u)}\abs{\sum_{j \in C_u}\Psi(u, j, t)w_{(u,\sigma^{-1}(u, v),j)}}^2}{\rho(u, t)} \label{eq:conv_comb0}
\end{align}
is exactly the value of the inner product $\ev{W^{\dagger}W}{\Psi(u, t)}$, with 
$$\ket{\Psi(u, t)} = \sum_{i \in C_u}\Psi(u, i, t)\ket{u, i}.$$
To see this, note that the correspondence given by the function $\sigma^{-1}(u, v)$ between degrees of freedom is unique, as well as the correspondence between the degrees of freedom of $u$ and its neighbors, yielding
\begin{align}
& \sum_{v \in N^{+}(u)} \abs{\sum_{j \in C_u}\Psi(u, j, t)w_{(u, \sigma^{-1}(u, v), j)}}^2 = \sum_{k \in C_u} \abs{ \sum_{j \in C_u}\Psi(u, j, t)w_{(u,k,j)} }^2.
\end{align}
Due to the unitarity of $W$, such inner product is precisely $\rho(u, t)$, proving property $2$.

Property $3$ follows trivially from the definition of the Markovian matrix $P$ in Equation \ref{eq:p_matrix} and from the orthogonality of the basis states.
\end{proof}

Theorem \ref{th:qw_locality} establishes that any discrete-time coined quantum walk with unitary operators $W$ and $S$, respectively described by Equation \ref{eq:coin_op} and Relation \ref{eq:shift_def}, is statistically equivalent, from the perspective of vertex probability evolution, to a non-homogeneous random walk over the same graph. Note that both $W$ and $S$ may depend on time, as long as the graph connectivity restrictions remain valid. 
\section{Generalization for multiple walkers} \label{sec:multiple}

To extend Theorem \ref{th:qw_locality} for multiple walkers, some additional definitions are needed.
In particular, the Hilbert space in which the process unfolds grows to allow for the joint description of the walkers. Let $K$ denote the number of walkers and, again, let $\hilb_w$ denote the Hilbert space for a single-walker on $G$. The enlarged space for $K$ walkers is $\hilb_w^{K} = \bigotimes_{i=1}^ {K} \hilb_{w}$. Let $\bm{v} = (v_1,...,v_k)$ denote an ordered sequence of $K$ vertices and $\bm{c} = (c_1,...,c_k)$ denotes its associated degrees of freedom such that $c_i \in C_{v_i}$. Let the set $B^{K} = \{\ket{\bm{v}, \bm{c}}\}$ denote a basis for $\hilb_{w}^{K}$ of which elements represents the joint position of the $K$ walkers. Let
$$\ket{\Psi(t)} = \sum_{\ket{\bm{v}, \bm{c}} \in B^{K}} \Psi(\bm{v}, \bm{c}, t) \ket{\bm{v}, \bm{c}}$$
denote the state of the system at instant $t$ and $\rho(\bm{v}, \bm{c}, t)$ be the joint probability distribution of states at instant $t$.
Assuming each walker can behave differently, with specific coins and shift operators, let $W_i$ and $S_i$ respectively denote the coin and shift operator for the $i\text{-th}$ walker, implying that the full operators are of the form $S = \bigotimes_{i=1}^{K}S_i$ and $W = \bigotimes_{i=1}^{K}W_i$.
If there is no interaction among the walkers, the system evolves, in the enlarged space, as in Equation \ref{eq:qwalk} and the joint distribution of vertices at an instant $t$ is merely $$\rho(\bm{v}, t) = \prod_{i=1}^ {K} \rho(v_i, t).$$

A more interesting scenario appears when the walkers can interact, allowing a dependency among the marginal probability distributions of the walkers. Let $U : \hilb_{w}^{K} \rightarrow \hilb_{w}^{K}$ be a unitary operator defined as 
\begin{equation}
    U = \sum_{\ket{\bm{v, c}} \in B^{K}} \sum_{\bra{\bm{v, c'}} \in B^{K\dagger}} \theta(\bm{v, c, c'})\dyad{\bm{v, c}}{\bm{v, c'}} \label{eq:inter}
\end{equation}
which accounts for walker interactions, such that the whole system state evolves as
\begin{equation}
    \ket{\Psi(t + 1)} = SWU\ket{\Psi(t)}. \label{eq:mult_evo}
\end{equation}

Under constraints of unitarity, the interactions performed by $U$ can be arbitrarily defined by specifying the values of $\theta(\bm{v, c, c'})$. Its inherent restriction resides on the self-mapping of the set of states that represent the vertex position for the $K$ walkers, as $\ket{\bm{v, c}}$ cannot be mapped to $\ket{\bm{u, c'}}$ for $\bm{u} \neq \bm{v}$. This mapping implies that $U$ does not move any of the walkers, confining movement to the action of the enlarged shift operator $S$. Nevertheless, diverse operations are allowed by $U$, such as generic controlled phase shifts, amplitude mixing and even amplitude shifts within the degrees of freedom of a walker controlled by the position of the others. Within this framework, the connectivity restrictions of the dispersion of the wavefunction are maintained, since amplitudes can only be transmitted through the edges of the graph. Theorem \ref{th:qw_mult_locality} follows as an extension of Theorem \ref{th:qw_locality} for this broader context, in which the movement of $K$ quantum walks is shown to be statistically equivalent to that of $K$ non-homogeneous random walks.

\begin{theorem}[Local convex evolution of multiple interacting walkers] \label{th:qw_mult_locality}
For any time instant $t$, the evolution of vertex probabilities for the $K$ walkers performed by the action of the unitary operator $SWU$ is locally convex, and is given by the Markovian matrix
\begin{equation}\label{eq:p_mult_matrix}
    p_{\bm{vu}}(t) =
    \begin{cases}
      \frac{\rho(\bm{v}, \bm{c}, t + 1)}{\rho(\bm{u}, t)}, & \text{if } \rho(\bm{u}, t) > 0 \text{ and } (u_i, v_i) \in E \text{ for all } $i$ \\
      \frac{1}{d(\bm{u})}, & \text{if } \rho(\bm{u}, t) = 0 \text{ and } (u_i, v_i) \in E \text{ for all } $i$ \\
      0, & \text{otherwise}
    \end{cases}
\end{equation}
when applied to the vertex distribution vector $\rho(t) \in \mathbb{R}^{|V|^k}_+$,
where $\bm{c} = \sigma(\bm{u},\bm{v})$ and $i \in \{1,...,K\}$, and such that $\rho(t + 1) = P(t)\rho(t)$.
\end{theorem}

\begin{proof}
The random walk dictated by matrix $P(t)$ accounts for the joint movement of the walkers in the sense that an index $\bm{u}$ of $P$ is a vector of dimension $K$ and denotes the position of the walkers. To formalize, let $G'=(V', E')$ denote a graph with $V' = V^{K}$ and $E' = E^{K}$, such that, for all $\bm{v}, \bm{u} \in V'$ with $\bm{v} = (v_1,...,v_K)$ and $\bm{u} = (u_1,...,u_K)$, $e = (\bm{v}, \bm{u}) \in E'$ if, and only if $(v_i, u_i) \in E$ for all $i$. Note that $d(\bm{u}) = \prod_{i = 1}^{K} d(u_i)$. In particular, each vertex of $G'$ represents the simultaneous position of all walkers and its edges codifies all of their possible combined movements.
It must be shown that $P(t)$ indeed represents a non-homogeneous random walk over $G'$ and that its vertex probability evolution matches Equation \ref{eq:mult_evo}.

The three properties which where shown to hold for Theorem \ref{th:qw_locality} are to be demonstrated for this general case, since the requirements for one walker extend to $K$ walkers naturally.
Let
\begin{align}
    & \ket{\Psi(\bm{u}, t)} = \sum_{\bm{c} \in C_{\bm{u}}}\Psi(\bm{u, c}, t)\ket{\bm{u, c}}
\end{align}
denote the overall state of $\bm{u} \in V'$ such that $\norm{\ket{\Psi(\bm{u}, t)}}^{2} = \rho(\bm{u}, t)$. Note that $\ket{\Psi(\bm{u}, t)} \in \hilb_w^{K}$, that $\bm{c} = (c_1,...,c_K)$ is a tuple denoting the degrees of freedom of each walker and that the functions $\eta$ and $\sigma$ are now defined for tuples of vertices and degrees of freedom. 
Assuming that $\bm{v} = \eta(\bm{u, c'})$ and $\sigma(\bm{u}, \eta(\bm{u, c'})) = \bm{c}$ for a given  $\bm{c'} \in C_{\bm{u}}$, the action of $SWU$ gives
\begin{align}
    & \frac{\rho(\bm{v, c}, t + 1)}{\rho(\bm{u}, t)} = \frac{\norm{\dyad{\bm{v,c}}{\bm{v,c}} SWU \ket{\Psi(\bm{u}, t)}}^{2}}{\norm{\ket{\Psi(\bm{u}, t)}}^{2}}. \label{eq:ratio_mult}
\end{align}
Since $SWU$ is unitary and $\braket{\bm{v, c}}{s} \leq 1$ for any unitary $\ket{s} \in \hilb_{w}^{K}$, the Inequality
\begin{align}
    & 0 \leq \frac{\norm{\dyad{\bm{v,c}}{\bm{v,c}} SWU \ket{\Psi(\bm{u}, t)}}^{2}}{\norm{\ket{\Psi(\bm{u}, t)}}^{2}} \leq 1 \label{eq:unit_ratio}
\end{align}
demonstrates property $1$.

Simultaneously, the action of $SWU$ also implies that the Inequality
\begin{align}
     & \dyad{\bm{v, c_{v}}}{\bm{v, c_{v}}} SWU \ket{\Psi(\bm{u}, t)} \neq 0 \label{eq:mult_nonzero}
\end{align}
is only valid for $\bm{v} \in V$ and $\bm{c_v} \in C_{\bm{v}}$ if $\bm{v} = \eta(\bm{u, c})$ and $\bm{c_{v}} = \sigma(\bm{u, c})$ for some $\bm{c} \in C_{\bm{u}}$. Assuming that $\bm{v} = \eta(\bm{u,c})$ and $\bm{c_{vu}} = \sigma(\bm{u, c})$, the last condition gives
\begin{align}
    & \sum_{\bm{v} \in N^{+}(\bm{u})} \frac{\norm{\dyad{\bm{v, c_{vu}}}{\bm{v, c_{vu}}} SWU \ket{\Psi(\bm{u}, t)}}^{2}}{\norm{\ket{\Psi(\bm{u}, t)}}^{2}} = 1. \label{eq:convex_mult}
\end{align}
Due to the fact that $SWU$ is a unitary operator, Equation \ref{eq:convex_mult} and the orthogonality of the basis states lead to properties $2$ and $3$.

\end{proof}

Essentially, Theorem \ref{th:qw_mult_locality} constructs a non-homogeneous random walk on $G'$ that matches the evolution of the joint vertex probability distribution of $K$ walkers induced by $SWU$ and, thus, asserts that the vertex probability distribution of the multiple walker interaction model has a local convex evolution on the vertices of $G'$.
Again, it is worth emphasizing that $SWU$ varies with time, as long as unitarity, graph connectivity and the conditions for the interaction operator $U$ remains valid.
\section{Simulation of quantum walk trajectories} \label{sec:simulation}

Theorems \ref{th:qw_locality} and \ref{th:qw_mult_locality} establish respectively the construction procedure for a non-homogeneous random walk which is statistically equivalent to any given single- and multiple-walker quantum walk. This random walk can be simulated to generate graph trajectories that capture the quantum walk behavior. The simulation of a random walk naturally constructs random paths on a graph. At each instant $t + 1$, the walker can only be found in an outward neighbor of node $v$, given that it was in node $v$ at instant $t$. Thus, the simulation constructs a sample path that ensures neighbor locality. We denote this sample path by \textit{quantum walk trajectory}. Note that this procedure is fundamentally different than the usual simulation procedure for quantum walks, where the distribution $\rho(t)$ is sampled independently at each time instant $t$ and no graph trajectory is constructed. To exemplify, Theorem \ref{th:qw_locality} was used to simulate quantum walks on a $2\text{-D}$ torus with Hadamard and Grover coins, and moving-shift operators, generating the ensembles of trajectories depicted in Fig.\ref{fig:trajectories}. Without loss of generality, the following discussion assumes a single-walker.

\begin{figure}
	\centering
    \begin{minipage}[b]{0.45\textwidth}
      \includegraphics[width=\textwidth]{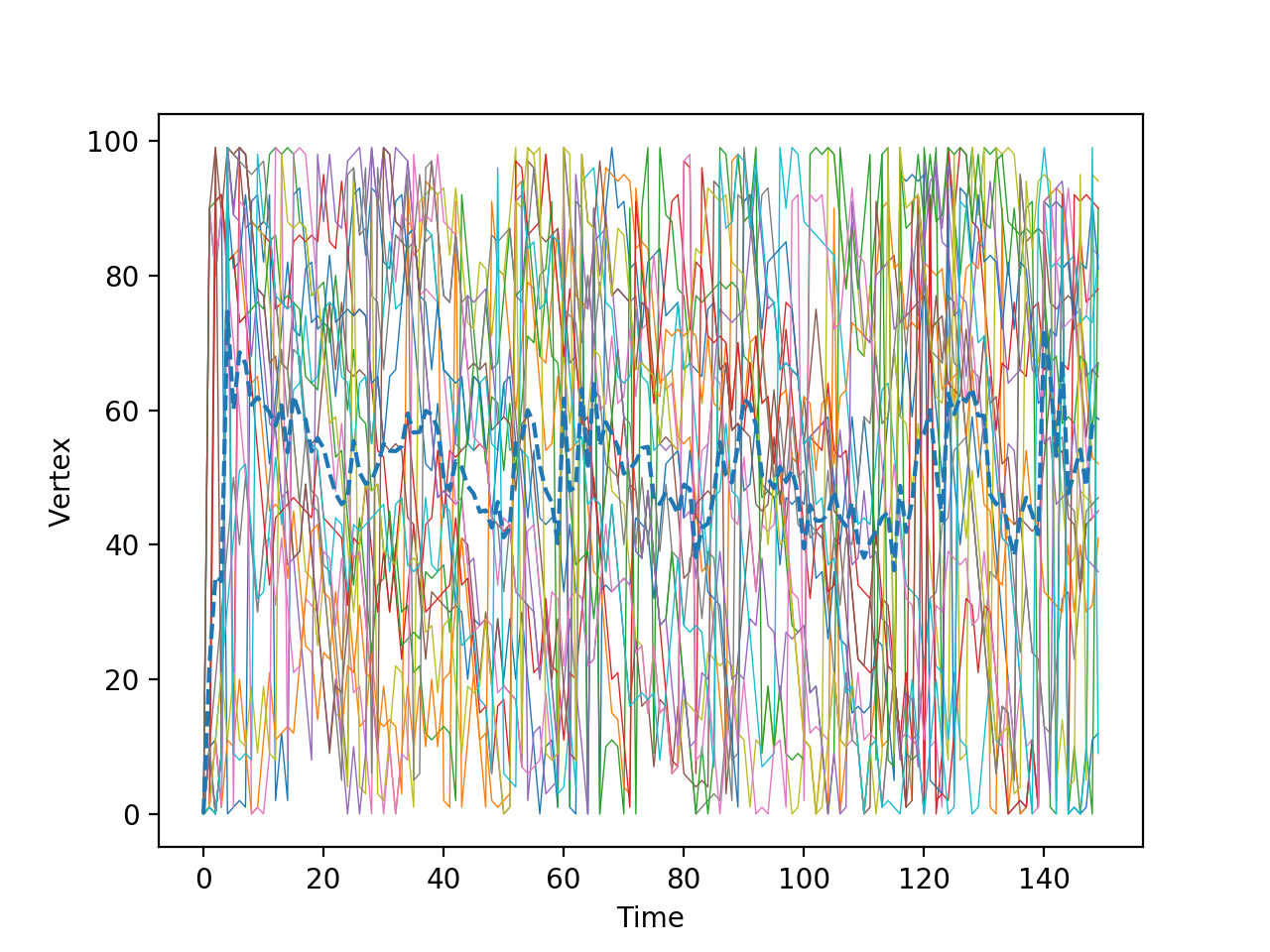}
    \end{minipage}
    \begin{minipage}[b]{0.45\textwidth}
	    \includegraphics[width=\textwidth]{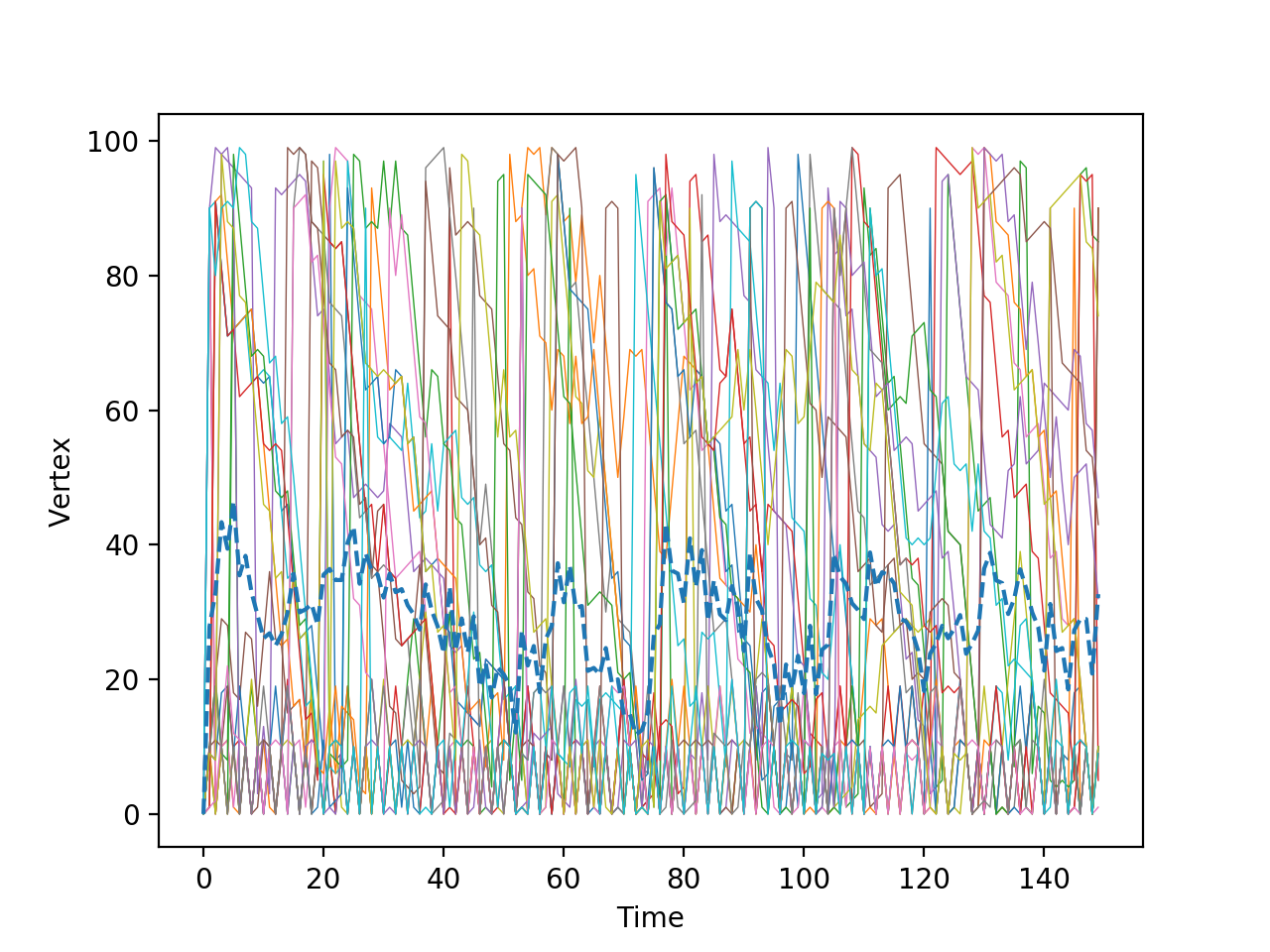}
    \end{minipage}
   	\caption{Ensembles of 20 trajectories obtained using Theorem \ref{th:qw_locality} for a Hadamard-coined (left) and Grover-coined (right) walks on a 10-by-10 2-D torus with moving shift operators for a localized initial state $\ket{0,0}$ (left edge of the origin). Each one of the solid lines corresponds to a given trajectory and the dashed line gives its empirical average of the vertex random variable per instant and indicate the difference between the dynamics induced by the different coin operators.}
    \label{fig:trajectories}
\end{figure}

While one simulated trajectory respects locality, an ensemble of trajectories recover the distribution of the quantum walk for every $t$. In particular, let $\chi = \{\tau_1, ..., \tau_M\}$ be an ensemble of $M$ independent trajectories. Let $\tau_i(t)$ denote the vertex visited by the walker at instant $t$ in the $i\text{-th}$ trajectory. Let $\mathbbm{1}(.)$ denote an indicator function activated by its argument condition. Let 
\begin{align}
    & \hat{p}_{u}^{M}(t) = \frac{1}{M}\sum_{i = 1}^{M} \mathbbm{1}(\tau_i(t) = u) \label{eq:rel_frac}
\end{align}
denote the fraction of time node $u$ was visited by the walker at instant $t$. Thus, by the law of large numbers, $\hat{p}_{u}^{M}(t) \xrightarrow{} \rho_{u}(t)$ as $M \xrightarrow{} \infty$ and the trajectories recover the node distribution of the quantum walk for all $t$.

Convergence is observed through the decreasing behavior of the total variation distance
\begin{align}
    & D_t(p, \rho) = \frac{1}{2} \sum_{v}\abs{\hat{p}_{v}^{M}(t) - \rho(v, t)} \label{eq:tvd}
\end{align}
between the empirical vertex distribution of the trajectory ensemble and the quantum walk vertex distribution, as it can be seen in Fig.\ref{fig:trajectories_tvd} for a Grover-coined quantum walk on the $2\text{-D}$ torus with moving-shift operator.

\begin{figure}
	\centering
    \begin{minipage}[b]{0.45\textwidth}
      \includegraphics[width=\textwidth]{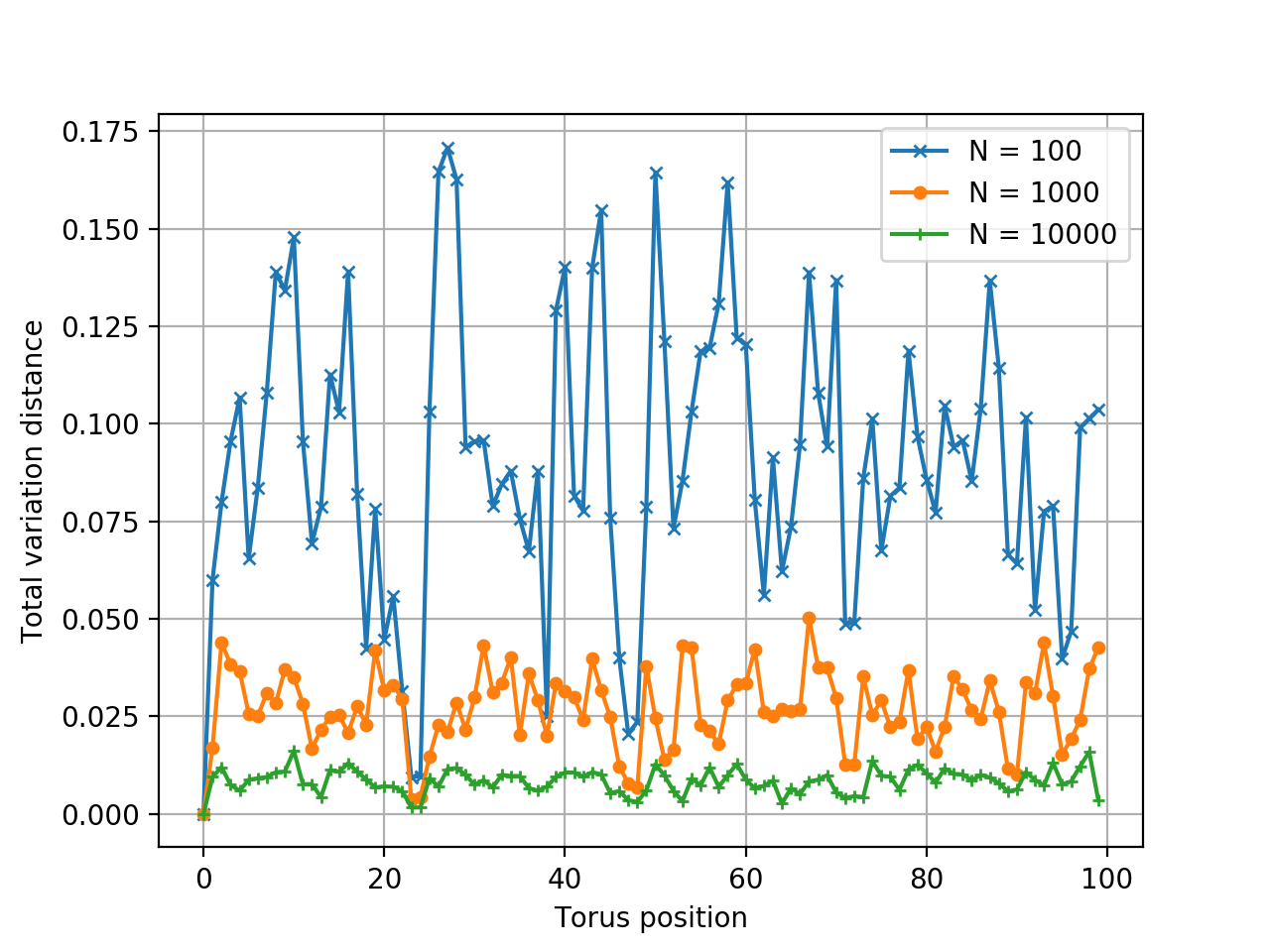}
    \end{minipage}
    \begin{minipage}[b]{0.45\textwidth}
	    \includegraphics[width=\textwidth]{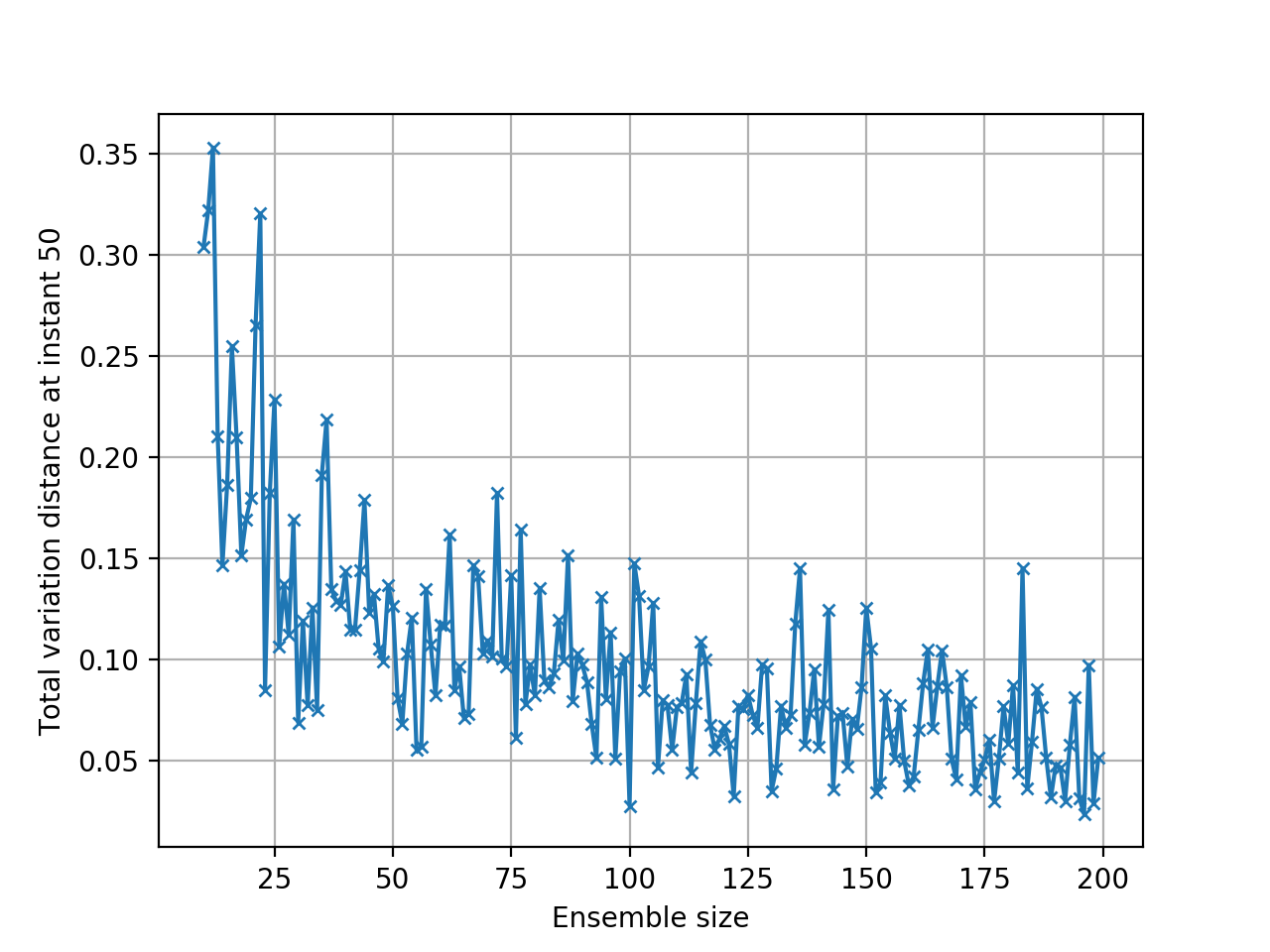}
    \end{minipage}
   	\caption{Analysis of the total variation distance between the empirical distribution obtained from trajectory sampling and the quantum walk distribution of a Grover-coined walk with moving shift on a 10-by-10 2-D torus. The left figure shows the total variation value for the distribution at each instant, with different curves representing different ensemble sizes. The right figure gives the distance value for a fixed time instant when the ensemble size grows.}
    \label{fig:trajectories_tvd}
\end{figure}

The non-homogeneous random walk simulation is a novel perspective for the study of quantum walks as it gives an efficient (polynomial) procedure for sampling trajectories which recover, by the law of large numbers, the vertex probability sequence of any quantum walk. As a matter of fact, measuring quantum walks on a possible physical implementation or independently sampling vertices from the quantum sequence of distributions do not address this question, since samples are obviously independent and there is no guarantee for obtaining trajectories. An alternative to sample trajectories would be to consider a rejection method that accepts only sequences of vertices that correspond to paths in the graph.
However, the marginal empirical distributions $\hat{p}_v(t)$ within the accepted trajectories would not necessarily match $\rho(v, t)$. See details in Appendix \ref{ap:rejection}.

\section{Complexity of random walk description and simulation} \label{sec:complexity}

An interesting question which arises once Theorems \ref{th:qw_locality} and \ref{th:qw_mult_locality} are considered is the computational complexity involved in constructing the corresponding random walk matrices $P(t)$. The information required to compute its entries at time instant $t$ are the probability distributions of the quantum walk at times $t$ and $t + 1$. Thus, if the state and the vertex distributions for time $t$ and $t + 1$ are known, constructing the matrix $P(t)$ has an intrinsic complexity of $O(|V|^2)$, since each of its entries can be computed in $O(1)$.

In general, however, the matrix can be computed by using Equation \ref{eq:qwalk} to calculate both $\ket{\psi(t)}$ and $\ket{\psi(t + 1)}$. Assuming that both $S$ and $W$ may vary with time, the cost for computing the wave function is $O(t|E|^2)$, since $t$ matrix-by-vector multiplications are performed, each with complexity $O(|E|^2)$. Computing the vertex distribution from the wavefunction has complexity $O(|E|)$, since the probability of each outward edge of a vertex must be considered. Hence, the overall complexity is $O(t|E|^2 + |E| + |V|^2) \in O(t|E|^2)$.

The computation of the quantum walk wavefunction is the general bottleneck for constructing the random walk matrices, unless the probability distributions of the quantum walk can be computed more efficiently. In terms of complexity, the problem of describing the probability evolution of the non-homogeneous random walk is at least as hard as solving the quantum walk distribution.
Nonetheless, specific walker systems can have their wavefunctions computed by algorithms that are more efficient than direct matrix multiplication. Walker dynamics with known closed-formula expressions for the wavefunction are an interesting case. For example, a generic coined quantum walk on an infinite line for which the walker moves in a single direction or remains on its position at every instant has known explicit probability distribution for all time $t$~\cite{Montero2015}.

Alternatively, for particular quantum walks and graphs, the wavefunction and the vertex probability distribution may be computed recursively and more efficiently than the general approach (see Appendix \ref{ap:grover} for an example).

\section{Conclusions} \label{sec:conclusions}

As the central contribution of this work, Theorems \ref{th:qw_locality} and \ref{th:qw_mult_locality} establish a construction procedure for non-homogeneous random walks that yield the same vertex probability distribution sequence of any single or multiple quantum walk. Besides establishing a formal equivalence between the two processes, this procedure allows for the efficient simulation of quantum walk trajectories, which can be used to investigate quantum walks from the perspective of vertex locality, as opposed to the simulation of independent samples over time. In a nutshell, the Theorems establish a formal correspondence between random and quantum walks on the same graph by showing that the vertex distribution of the two processes are identical for all time $t$. Moreover, any statistical property of a quantum walk can be analyzed through quantum walk trajectories. This concept and its simulation could possibly be exploited to evaluate theoretical properties and concepts of Markov chains which were initially modified to address quantum walks, such as mixing, dispersion and hitting time \cite{aharonov2001quantum}, in their original circumstances.

Due to the universality of quantum walks for quantum computation, both Theorems may have important implications in the development of this larger field.The connections between generic computational processes and time-dependent Markov chains can be explored to guide new interesting research on quantum computing. On the other hand, Theorems \ref{th:qw_locality} and \ref{th:qw_mult_locality} do not provide improvements on the computation of the state probability distribution sequences of quantum walks, since the construction of the non-homogeneous random walk requires solving the quantum problem.


While this work showed that any single or multiple quantum walk has a corresponding random walk, an interesting future consideration would be establishing the reverse correspondence. In particular, answering whether or not any (single and multiple) random walk has a corresponding quantum walk.
\begin{appendices}
\section{Limitations of the rejection method for QWT} \label{ap:rejection}
Assume that quantum walk trajectories of length $L$ are to be sampled, such that the marginal vertex probability within the trajectories, $p(v, t)$, is exactly $\rho(v, t)$, for $v \in V$ and $t \in \{0,...,L - 1\}$.
Let $\mathcal{X}$ and $\mathcal{T}$ denote the set of all possible sequences of measurements and graph trajectories with length $L$, respectively. Let $\mathcal{X}_{v}^{T} \subset \mathcal{X}$ and $\mathcal{T}_{v}^{t} \subset \mathcal{T}$ denote the set of all sequences and trajectories of length $L$ in which $v$ appears in position $t$. It follows trivially from independence that the probability of a sequence of measurements $\tau \in \mathcal{X}$ is
\begin{align}
    & p(\tau) = \prod_{i=1}^{L}\rho(\tau_t),
\end{align}{}
where $\tau_t$ denotes the vertex measured at $t$. In this case, the vertex probability at $t$ is simply, 
\begin{align}
    & \rho(v, t) = \sum_{\tau \in \mathcal{X}_{v}^{t}} p(\tau) \label{eq:rho_prob}.
\end{align}
However, rejecting non-trajectory samples yields
\begin{align}
    &  p(v, t) = \frac{\sum_{\tau \in \mathcal{T}_{v}^{t}} p(\tau)} {\sum_{\tau^{'} \in \mathcal{T}} p(\tau^{'})}. \label{eq:rejec_prob}
\end{align}
It is not clear whether Equations \ref{eq:rho_prob} and \ref{eq:rejec_prob} are equal for every possible quantum walk, since the ratio between the probability of generating trajectories with vertex $v$ at position $t$ and the probability of constructing a trajectory would have to be to $\rho(v, t)$, for all $v \in V$ and $t \in \{0,..., L - 1\}$.

Additionally, even for the cases where Equations \ref{eq:rho_prob} and \ref{eq:rejec_prob} are equal, the expected time to accept a sample in the rejection procedure is precisely the inverse of the probability of generating a trajectory. Although this probability depends on a myriad of factors, which brings difficulties to a general analytical evaluation of the sampling efficiency, the number of trajectories of a given length $L$ within a graph can be exponentially smaller than the number of possible sequences of measurements. As an example, a $D\text{-dimensional}$ torus with $V$ vertices would have $VD^{L - 1}$ paths of length $L$ and $V^{L}$ possible sequences, which for values of $D \in O(1)$ is exponentially larger than the number of trajectories, suggesting that the expected time to generate a trajectory sample would be unfeasible.

In precise terms, Theorem \ref{th:qw_locality} offers a polynomial procedure to sample graph trajectories for any quantum walk. Given the transition matrix $P(t)$, a quantum trajectory of lenght $L$ can be sampled in time $O(Ld_{max})$ where $d_{max}$ is the maximum degree of the graph. \footnote{The alias method could also be used if multiple samples are to be generated in which case the amortized time complexity for the trajectories is $O(L)$.}


\section{Dynamic programming for Grover walk on torus} \label{ap:grover}
The Grover-coined (Eq.\ref{eq:grover_op}) walk on a $D\text{-dimensional}$ torus with moving-shift operator (Eq.\ref{eq:moving_shift}) and purely real initial conditions serves as an example where the probability distribution can be computed by a dynamic programming algorithm that is more efficient than direct matrix multiplication. The number of degrees of freedom within the $D\text{-dimensional}$ torus is $2D$. Analyzing the action of the total walk operator $MG$ on a given state
$$\Psi(u, t) = \sum_{c \in C_u} \Psi(u, c, t) \ket{(u,c)}$$
and having $\eta(u, c) = v$, the probability $\rho(v, c, t+1)$ is described as
\begin{align}
    & \rho(v, c, t + 1) = \abs{\frac{\sum_{c' \in C_u} \Psi(u, c', t)}{D} - \Psi(u, c, t)}. \label{eq:gr_tor1}
\end{align}

Assuming that $\psi(u,c,t) = \sqrt{\rho(u,c,t)}e^{i\cos{\theta_{uct}}}$ and noting that $\theta_{uct} = 0$ for every $u$, $c$ and $t$ whenever purely real initial conditions are considered yields
\begin{align}
     & \rho(v, c, t + 1) = \abs{\frac{\sum_{c' \in  C_u} \sqrt{\rho(u, c', t)}}{D} - \sqrt{\rho(u, c, t)}}. \label{eq:gr_tor2}
\end{align}
From Theorem \ref{th:qw_locality}, the entries of the random walk matrices for which $\rho(v, c, t) > 0$ are given by 
\begin{align}
    & \frac{\rho(v, c, t + 1)}{\rho(u, t)} = \frac{1}{\rho(u, t)} \abs{\frac{\sum_{c' \in  C_u} \sqrt{\rho(u, c', t)}}{D} - \sqrt{\rho(u, c, t)}} \label{eq:gr_tor3}.
\end{align}{}


Equation \ref{eq:gr_tor3} can be solved through a dynamic programming algorithm in which each time instant has complexity $O(|V|^2D)$, implying on an overall procedure of complexity $O(t|V|^2D)$ for all matrices up to time $T$. For $D \in O(1)$, the algorithm has complexity $O(t|V|^2)$, showing a quadratic improvement over the generic procedure, since $O(t|E|^2) \in O(t|V|^4)$.
\footnote{$D$ can be at most $|V|$ for the case when the torus degenerates to the complete graph. The overall procedure for this case is in $O(t|V|^3)$, which still represents a linear improvement from the general procedure of matrix multiplication.}



\end{appendices}{}

\bibliography{references}

\end{document}